\documentclass[letterpaper, 10 pt, conference]{ieeeconf}  

\IEEEoverridecommandlockouts       

\usepackage{macros}
\usepackage{mathtools}
\usepackage{float}
\DeclarePairedDelimiter{\ceil}{\lceil}{\rceil}

\title{\LARGE \bf
Generalized Deterministic Perturbations For Stochastic Gradient Search
}

\author{Chandramouli K.$^1$,
        Prabuchandran K.J.$^1$$^2$,
        D. Sai Koti Reddy$^3$, 
        and Shalabh Bhatnagar$^1$$^4$ 
\thanks{$^1$ Department of Computer Science and Automation, Indian Institute of Science (IISc)} 
\thanks{$^2$ Supported by Amazon-IISc Postdoctoral fellowship}
\thanks{$^3$ IBM Research, Bangalore}
\thanks{$^4$ Robert Bosch Centre for Cyber-Physical Systems, IISc}}

\begin{document}
\maketitle
\begin{abstract}
Stochastic optimization (SO) considers the problem of optimizing an objective function in the presence of noise. Most of the solution techniques in SO estimate gradients from the noise corrupted observations of the objective and adjust parameters of the objective along the direction of the estimated gradients to obtain locally optimal solutions. Two prominent algorithms in SO namely Random Direction Kiefer-Wolfowitz (RDKW) and Simultaneous Perturbation Stochastic Approximation (SPSA) obtain noisy gradient estimate by randomly perturbing all the parameters simultaneously. This forces the search direction to be random in these algorithms and causes them to suffer additional noise on top of the noise incurred from the samples of the objective. Owing to this additional noise, the idea of using deterministic perturbations instead of random perturbations for gradient estimation has also been studied. Two specific constructions of the deterministic perturbation sequence using lexicographical ordering and Hadamard matrices have been explored and encouraging results have been reported in the literature. In this paper, we characterize the class of deterministic perturbation sequences that can be utilized in the RDKW algorithm. This class expands the set of known deterministic perturbation sequences available in the literature. Using our characterization, we propose construction of a deterministic perturbation sequence that has the least cycle length among all deterministic perturbations. Through simulations we illustrate the performance gain of the proposed deterministic perturbation sequence in the RDKW algorithm over the Hadamard and the random perturbation counterparts. We also establish the convergence of the RDKW algorithm for the generalized class of deterministic perturbations. 
\end{abstract}

\section{Introduction}
Stochastic optimization (SO) problems frequently arise in engineering disciplines such as transportation systems, machine learning, service systems, manufacturing etc. Practical limitations, lack of model information and the large dimensionality of these problems prohibit analytic solutions to these problems. Simulation is often employed to evaluate the performance of the current parameters of the system. Simulating and evaluating the system's performance is generally expensive and one is typically constrained by a simulation budget. In such scenarios, owing to the simulation budget one aims to drive the system to optimal parameter settings using as few simulations as possible.

Under the SO framework, we have a system that gives noise-corrupted feedback of the performance for the currently set parameters, i.e., given the system parameter vector $\theta$, the feedback that is available is the noisy evaluation $h(\theta, \xi)$ of the performance $J(\theta)=\E_{\xi}[h(\theta, \xi)]$ where $\xi$ is the noise term inherent in the system and $J(\theta)$ denotes the expected performance of the system for the parameter $\theta$. The pictorial description of such a system is shown in Figure \ref{fig:so}. The objective in the SO problem then is to determine a parameter $\theta^*$ that gives the optimal expected performance of the system, i.e., 
\begin{align}
 \theta^* = \arg\min_{\theta \in \R^p} J(\theta). \label{eq:pb}
\end{align}

\tikzstyle{block} = [draw, fill=white, rectangle,
   minimum height=3em, minimum width=6em]
\tikzstyle{sum} = [draw, fill=white, circle, node distance=1cm]
\tikzstyle{input} = [coordinate]
\tikzstyle{output} = [coordinate]
\tikzstyle{pinstyle} = [pin edge={to-,thin,black}]
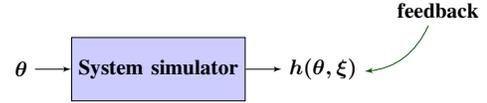
\begin{figure}[t]
  \centering
\scalebox{0.8}{\begin{tikzpicture}[auto, node distance=2cm,>=latex']
\node (theta) {$\boldsymbol{\theta}$};
\node [block, fill=blue!20,right=0.6cm of theta,align=center] (sample) {\makecell{\textbf{System simulator}}}; 
\node [right=0.6cm of sample] (end) {$\boldsymbol{{h(\theta, \xi)}}$};
\node [ above right= 0.6cm of end] (bias){\textbf{feedback}};
\draw [->] (theta) --  (sample);
\draw [->] (sample) -- (end);
\path [darkgreen,->] (bias) edge [bend left] (end);
\end{tikzpicture}}
\caption{Stochastic Optimization Model}
\label{fig:so}
\end{figure}

Analogous to solutions for deterministic optimization problems where the explicit analytic gradient of the objective function is used to adjust the parameters along the negative gradient directions, many of the solution approaches in SO mimic the familiar gradient descent algorithm. However, unlike the deterministic setting, the SO setting only has access to noise corrupted samples of the objective. Thus, in the SO setting, one essentially aims at estimating the gradient of the objective function using noisy cost samples. In the pioneering work by Kiefer and Wolfowitz \cite{kiefer1952}, the  gradient is estimated by approximating each of the partial derivatives using either a two-sided or a one-sided finite difference approximation (FDSA) algorithm. This algorithm requires $2p$ objective function evaluations (or simulations) per iteration for the two-sided gradient approximation scheme and $p+1$ simulations per iteration for the one-sided scheme (for a $p$-dimensional parameter problem, see \cite{bhatnagar-book}). As the number of simulations per iteration required for gradient estimation scales linearly with the dimension of the problem, FDSA algorithm is expensive to deploy under high-dimensional parameter settings. 

In \cite{kushcla}, Random Direction Kiefer-Wolfowitz (RDKW) algorithm that uses only two simulations per iteration for obtaining gradient estimates has been proposed. In the RDKW algorithm, all the parameters are randomly perturbed simultaneously using two parallel simulations and function evaluations at those perturbed parameters are used to obtain the gradient estimate. In the RDKW algorithm, the random perturbation vector as well as the random direction vector involved in estimating the gradient have been kept the same. For the choice of random direction (or perturbation), various distributions like spherical uniform distribution \cite{kushcla}, uniform distribution \cite{ermol1969method}, normal and Cauchy distribution \cite{styblinski1990experiments}, asymmetric Bernoulli \cite{prashanth2017adaptive} have been explored. The number of simulations required for estimating the gradients in the RDKW algorithm is significantly less compared to the FDSA algorithm and the algorithm is seen to perform empirically better than FDSA.	

In a seminal work \cite{spall}, the Simultaneous Perturbation Stochastic Approximation (SPSA) algorithm that uses two simulations similar to RDKW has been proposed. Unlike the RDKW algorithm, SPSA employs different choices for parameter perturbations and the random direction of movement, in particular, the random perturbation direction and the random direction of movement have been chosen to be inverses of each other. In \cite{spall}, symmetric Bernoulli distribution has been shown to be the best choice for random perturbations among all the distributions and the proposed SPSA scheme has been proven to perform asymptotically better compared to FDSA. In \cite{chin1997comparative}, a comprehensive comparative study of the stochastic optimization algorithms namely FDSA, RDKW and SPSA has been provided. Further, under a general third order cross derivative assumption on the loss function, RDKW with symmetric Bernoulli distribution has been shown to be the best choice for random directions. In \cite{theiler2006choice}, an example of a loss function that does not satisfy the third order cross derivative condition in \cite{chin1997comparative} has been constructed. For such a loss function, it has been shown that the optimal distribution choice for random directions need not be symmetric Bernoulli.

In \cite{kushcla} and \cite{spall1997one}, to further reduce simulation cost per iteration, extensions of the RDKW and SPSA algorithms that estimate the gradient with only one simulation or measurement of the objective have been considered. However, it is observed that the one-simulation gradient estimate has higher bias compared to the two-simulation gradient estimate. In \cite{sandilya1997deterministic} and \cite{wang1998deterministic}, deterministic conditions for the perturbation and noise sequences required to obtain almost sure convergence of the iterates have been discussed. In \cite{bhatfumarcwang}, to enhance the performance of one-sided SPSA scheme, deterministic perturbations based on lexicographical ordering and Hadamard matrices have been proposed. Further, the numerical results in \cite{bhatfumarcwang}, illustrate the benefit of Hadamard matrix based perturbation sequences as it has been shown to improve the performance of SPSA empirically for the case of one sided measurements. In \cite{xiong2002randomized}, a unified view of both RDKW and SPSA is presented and a binary deterministic perturbation sequence using orthogonal arrays \cite{hedayat1999orthogonal} for obtaining gradient estimate in both of the algorithms has been discussed. 

In this paper, we generalize the class of deterministic perturbation sequences that can be utilized in the RDKW algorithm. Based on this characterization, we provide a construction of a deterministic perturbation sequence using a specially chosen circulant matrix. We empirically study the performance of the constructed sequence against the afore mentioned Hadamard matrix based deterministic perturbations and the randomized perturbations. We expect with our generalization the study of rate of convergence for the RDKW algorithm based on deterministic perturbation sequences would be possible. We now summarize our contributions:
\begin{itemize}
\item We generalize the class of deterministic perturbation sequences that can be applied in the RDKW algorithm.
\item We provide a special construction of deterministic perturbation sequence with smaller cycle length compared to Hadamard perturbation sequence.  
\item We illustrate the performance gain of the proposed deterministic perturbations over the Hadamard matrix based perturbations as well as random perturbations.
\item We prove the convergence of the RDKW algorithm for the class of deterministic perturbations. 
\end{itemize}


\section{Conditions on Deterministic Perturbations}
In this section, we describe the classical RDKW algorithm and motivate the 
necessary conditions that a deterministic perturbation sequence should satisfy for 
almost sure convergence of the iterates in the deterministic perturbation version of 
RDKW algorithm.

The standard RDKW algorithm iteratively updates the parameter vector along the direction of the negative estimated gradient, i.e.,
\begin{align}
\label{eq:grad-descent}
\theta_{n+1} = \theta_n - a_n \widehat{\nabla J}(\theta_n), 
\end{align}
where $a_n$ is the step-size that satisfies standard stochastic approximation conditions (see Assumption \textbf{A2} in section \ref{sec:convergenceresults}) and $\widehat{\nabla J}$ is the estimate of the gradient of the objective function $J$ at the current parameter. 

In the case of two-simulation RDKW algorithm, the gradient estimate at $\theta$ is obtained as
\begin{align}\label{gradEst1}
 \widehat{\nabla J}(\theta) = \frac{J(\theta + \delta d)-J(\theta-\delta d)}{2\delta}d,
\end{align}
where $d$ is the random perturbation direction chosen according to a specific probability distribution. The properties that the specific distribution on $d$ should satisfy can be obtained as explained below. The Taylor series expansion of $J(\theta \pm \delta d)$ around $\theta$ is given by
 \begin{equation} \label{eq:pmTaylor}
 J(\theta \pm \delta d)=
 J(\theta) \pm \delta d^T \nabla J(\theta)+o(\delta^2).
\end{equation}
From \eqref{eq:pmTaylor}, the error between the estimate and the true gradient at $\theta$ can be obtained as
\begin{align}\label{eq:Taylor}
&\frac{J(\theta +\delta  d )-J(\theta -\delta  d )}{2\delta }d 
 - \nabla J(\theta ) \nonumber \\
&= (d d ^T-I)\nabla J(\theta )+o(\delta ).
\end{align}
Note that the term $(dd ^T-I)\nabla J(\theta )$ constitutes the bias in the gradient estimate. For the error estimate in \eqref{eq:Taylor} to be negligible, we require 
\begin{equation}\label{eq:2sidedbias}
\mathbb{E}\Big{[}dd^T \Big{]}=I.
\end{equation}
Here, the expectation $\mathbb{E}[\cdot]$ is taken over the random perturbation distribution.

In the one-simulation version of the RDKW algorithm, the gradient estimate at $\theta$ is obtained as
\begin{align}\label{gradEst2}
 \widehat{\nabla J}(\theta) = \frac{J(\theta + \delta d)}{\delta}d.
\end{align}
By analogous Taylor series argument, we obtain the error between the estimate and the true gradient as
\begin{align}\label{eq:1sided}
&\frac{J(\theta +\delta  d )}{\delta }d 
- \nabla J(\theta ) \nonumber \\
&=\frac{J(\theta )}{\delta }d 
+(d d ^T-I)\nabla J(\theta )+O(\delta ).
\end{align}
From \eqref{eq:1sided}, we require the following to hold in addition to \eqref{eq:2sidedbias} in the case of random perturbations for the one simulation version of RDKW algorithm, i.e.,
\begin{equation}\label{eq:1sidedbias}
\mathbb{E}[d]=0. 
\end{equation}
For the random perturbations, $d \sim F$, $F$ is any distribution that satisfies \eqref{eq:2sidedbias} and \eqref{eq:1sidedbias}, the noise in the gradient estimates gets averaged asymptotically. An example distribution for $F$ is symmetric Bernoulli where each component of the perturbation vector is $\pm 1$ with equal probability.

From \eqref{eq:2sidedbias} and \eqref{eq:1sidedbias} clearly one is motivated to look for perturbations that satisfy similar properties. In what follows, the sequence of deterministic perturbations (that will be used in either \eqref{gradEst1} or \eqref{gradEst2}) will be denoted by $\{d_n\}_{n\geq 1}$ and we require the following two properties to hold for the perturbation sequence $d_n$ for the almost sure convergence of the iterates to a  local minima.
\begin{enumerate}[label=\textbf{P\arabic*.}]
  \item Let $D_n:=d_nd_n^T-I_{p \times p}.$ For any $s \in \mathbb{N}$ there exists a 
  $P \in \mathbb{N}$ such that $\sum\limits_{n=s+1}^{s+P}D_n=0$ and,
  \item  $\sum\limits_{n=s+1}^{s+P}d_n=0.$
\end{enumerate}
\begin{remark}\label{rem1}
The properties $\textbf{P1}$ and $\textbf{P2}$ are the deterministic analogues of \eqref{eq:2sidedbias} and \eqref{eq:1sidedbias}. For the properties $\textbf{P1}$ and $\textbf{P2}$ to hold, it is sufficient to determine a finite sequence $\{d_1,d_2,\dots,d_{P}\}$ such that $\sum_{n=1}^{P} d_n d_n^T = PI$ and $\sum_{n=1}^{P} d_n=0$ and for $n \geq P+1$, periodically cycle through this sequence, i.e., set $d_n=d_{n \% P + 1 }$. We will refer the length of the deterministic perturbation sequence $P$ as the cycle length.
\end{remark}

\section{Construction Of Deterministic Perturbations}\label{sec3}
In section \ref{detPerturb}, following Remark \ref{rem1}, we first characterize the finite sequences $\{d_1,d_2,\dots,d_{P}\}$ that satisfy properties $\textbf{P1}$ and $\textbf{P2}$ by providing a matrix equation whose solution gives the deterministic perturbations. In Section \ref{construct}, we then construct a specific sequence using a circulant matrix that has the least possible cycle length among all the deterministic perturbation sequences. Finally in section \ref{gradEstSec}, we completely describe the RDKW algorithm that uses the deterministic perturbation sequence constructed using the circulant matrix approach. 
\subsection{Matrix condition for Deterministic Perturbations}\label{detPerturb}
The properties \textbf{P1} and \textbf{P2} can be satisfied individually. For example, to satisfy property \textbf{P1}, let $P=p$ and $d_{n}=\sqrt{p} e_{n}, ~n \in \{1,\ldots,P\}$, the scaled canonical basis vectors, then $\sum_{n=1}^{P} d_{n}d_{n}^T = \sum_{n=1}^{p} pe_{n}e_{n}^T = pI$. To satisfy property \textbf{P2}, consider any set of linearly dependent vectors $\{v_0,\cdots,v_P\}$. Then there exists scalars $\alpha_1,\cdots, \alpha_P$ such that $\sum_{n=1}^{P} \alpha_n v_n=0$. Now for the choice $d_n=\alpha_nv_n$ the property \textbf{P2}, $\sum_{n=1}^{P} d_n=\sum_{n=1}^{P} \alpha_nv_n=0$ is trivially satisfied. A natural question would be to determine sequences $\{d_n\}_{1 \leq n \leq P}$ that satisfy both the properties simultaneously.

To address this problem, let us consider a $p \times P$ matrix $Y$ as follows:
$Y:=\left[\begin{array}{cccc}
\uparrow \ &\uparrow \ &\cdots \ &\uparrow \\
d_1 \  &d_2 \ &\cdots \ &d_{P} \\
\downarrow \ &\downarrow \ &\cdots \ & \downarrow \\
\end{array}\right].$
Let $u=[1,1,\cdots,1]^T$ be a $P \times 1$ dimension vector. The perturbations that satisfy properties \textbf{P1} and \textbf{P2} essentially solve the two matrix equations $Yu=0$ and $YY^{T}=PI$. These equations can be compactly written in a single matrix equation as 
\begin{align}\label{single}
XX^{T}=PI_{(p+1)\times(p+1)},
\end{align}
where $X=\left[\begin{array}{c}
u^{T}\\ Y \end{array}\right]$. Note that $Y_{p\times P}$ and $P$ are the unknowns here.

It can observed from \eqref{single} that $\frac{X}{\sqrt{P}}$ could be treated as a $p \times P$ submatrix of a $P \times P$ orthogonal matrix with the first row being $\frac{u^T}{\sqrt{P}}$, a $1 \times P$ vector. It has been shown in \cite{bhatfumarcwang} that columns of Hadamard matrices satisfy properties \textbf{P1} and \textbf{P2} simultaneously with $\bar{P}=2^{\log_2 \ceil{p+1}}$, i.e., $X$ is chosen as a $(p+1) \times 2^{\log_2 \ceil{p+1}}$ submatrix of the  
Hadamard matrix. It is not in general clear if the equation \eqref{single} can be solved for a smaller $P \leq \bar{P}$.
\begin{remark}\label{rem2}
We note that similar analysis for matrix condition for the construction of deterministic perturbations for SPSA estimates involves solving the following matrix system.
$AB=PI$,$Au=0$ and $A \circ B^T=vu^T$ where $A$ is $p\times P$, $B$ is $P\times p$, $u$ is $P \times 1$ vector of ones, $v$ is $p \times 1$ vector of ones and $\circ$ denotes the Hadamard product of the matrices $A$ and $B$. It is not clear how to solve for $P,$ $A$ and $B$ due to the presence of Hadamard product in this system. 
\end{remark}

\subsection{Specific Perturbation Sequence Construction}\label{construct}
In this section, our goal is to obtain a sequence with least cycle length. Using a simple matrix rank argument it can be shown that $P$ is at least $p+1$. Thus, in what follows, we give a construction of deterministic perturbation sequence with cycle length $P=p+1$.
We first write 
$$Y=\left[\begin{array}{ccc}
\uparrow  \ &\cdots \ &\uparrow \\
Z \ & \ &-ZU \\
\downarrow \  &\cdots \ & \downarrow \\
\end{array}\right]$$
where $Z$ is a $p\times p$ matrix and $U$ is any $p \times (P-p)$ matrix with columns that sum to 1. Clearly $Yu=0$ satisfies property \textbf{P2}.

To satisfy property \textbf{P1}, i.e., $YY^{T}=I$ is equivalent to
\begin{align}\label{mainEq}
ZZ^{T}+ZUU^{T}Z^{T}=Z(I+UU^{T})Z^{T}=PI.
\end{align}
Clearly construction of deterministic perturbations with smaller cycle length $P$ is equivalent to solving for
$Z$ with an appropriate choice of $U$.

The simplest choice of $U$ with column sums being 1 is $U=u$, a $p \times 1$ vector, thus $P=p+1$. Let $C=I+UU^{T}=I+uu^{T}$ ($p \times p$ dimensional matrix) 
\begin{equation}\label{eq:C}
C = \left[\begin{array}{cccc}
2 \ 1 \ 1 \cdots 1\\ 
1 \ 2 \ 1 \cdots 1 \\
\vdots \ \vdots \ \vdots \ \vdots\\
1 \ 1 \ 1 \cdots 2
\end{array}\right].
\end{equation}
Observe that $C$ is a positive definite circulant matrix. Hence $C^{-1/2}$ is well defined and the choice $Z=C^{-1/2}$ satisfies \eqref{mainEq} and solves the system $YY^{T}=I$ with  $P=p+1$, i.e., 
\begin{align}\label{yEqn}
Y=\sqrt{p+1}[C^{-1/2},-C^{-1/2}u].
\end{align}
The columns of $Y$ finally give us the deterministic perturbations. We note that in general the computation of $C^{-1/2}$ is $O(p^3)$ and can be very expensive for large $p$. However owing to the special structure of $C$,  using a Sherman-Morrison type result (see Lemma \ref{lemma: gen Sherman-Morrison}, Section \ref{sec:convergenceresults}), $C^{-1/2}$ can be computed in $O(p^2)$ time complexity.

\subsection{Gradient estimation}\label{gradEstSec}
In this section, we present the RDKW algorithms that use the deterministic perturbation sequence constructed above in two-simulation and one-simulation gradient estimates of the objective. We denote the corresponding algorithms by DSPKW-2C and DSPKW-1C respectively.  

\label{sec:algo}
\begin{algorithm}[H]
\begin{algorithmic}[1]
\State {\bf Input:}
\begin{itemize}
 \item $\theta_0 \in \mathbb{R}^p,$ initial parameter vector
 \item $\delta_n, n \geq 0,$ a sequence of sensitivity parameters to approximate gradient
 \item Matrix of perturbations $$Y=\sqrt{p+1}[C^{-1/2},-C^{-1/2}u],$$ 
 with $u=[1,1,\cdots,1]^T;$
 \item noisy measurements of cost objective $J$
 \item $a_n, n \geq 0,$ step-size sequence satisfying assumption \ref{stepsizes} (see section \ref{sec:convergenceresults})
 \item $n_{end}$, the total number of iterations determined by simulation budget
\end{itemize}
\State {\bf Output:} $\theta_{n_{\text{end}}}$, approximate local optimal solution
\For{$n = 1,2,\ldots n_{\text{end}}$}	
	\State Let $d_n$ be the mod$(n,p+1)^{\text{th}}$ column of $Y$. 
	\State Update the parameter as follows:
  \begin{equation*}
  \theta_{n+1}=\theta_n-a_n \widehat{\nabla J}(\theta_n)
  \end{equation*}
$\widehat{\nabla J}(\theta_n)$ is chosen according to either \eqref{eq:grad-twosided} 
or \eqref{eq:grad-onesided} for DSPKW-2C and DSPKW-1C respectively. 
\EndFor
\State {\bf Return} $\theta_{n_{\text{end}}}$
\end{algorithmic}
\caption{Basic structure of DSPKW.}
\label{alg:structure}
\end{algorithm}

Let $\delta_n, n\geq 0$ denote a sequence of diminishing positive real numbers satisfying assumption \ref{stepsizes} in section \ref{sec:convergenceresults}.
Let $y_{n}^{+}$, $y_{n}^{-}$ denote the noisy objective function evaluations at the perturbed parameters $\theta_n+\delta_n d_n$ and $\theta_n -\delta_n d_n$ respectively, i.e.,
$y_{n}^{+} = J(\theta_n+\delta_n d_n) + M_{n+1}^{+}$ and $y_{n}^{-} = J(\theta_n-\delta_n d_n) + M_{n+1}^{-}$. We assume the noise terms $M_{n}^{+}, M_{n}^{-}$ are martingale difference noise sequence, $\E\left[M_{n+1}^{+} | \F_n \right] = \E\left[ M_{n+1}^{-} | \F_n\right] = 0$ where $\F_n = \sigma(\theta_m, M^{+}_{m}, M^{-}_{m}, ~m\le n)$ is the information conditioned on the past parameter values and martingale difference terms.

The two-simulation and one-simulation estimates of the gradient $\nabla J(\theta_n)$ based on the observed noisy objective samples for the RDKW algorithm are respectively given by
\begin{align}
\label{eq:grad-twosided}
\widehat{\nabla J}(\theta_n)=
\left[\dfrac{(y_{n}^{+} - y_{n}^{-})d_n}{2\delta_n}\right],
\end{align}
\begin{align}
\label{eq:grad-onesided}
\widehat{\nabla J}(\theta_n)=
\left[\dfrac{(y_{n}^{+})d_n}{\delta_n}\right],
\end{align}
respectively. Observe that in the two-sided estimate \eqref{eq:grad-twosided} we use two function samples $y_{n}^{+}$ and $y_{n}^{-}$ 
and the estimate in \eqref{eq:grad-onesided} uses only one function sample $y_{n}^{+}$.

Now we briefly describe the DSPKW algorithm. Inputs to the DSPKW algorithm are randomly chosen initial point $\theta_0$,  diminishing sequences $\delta_n$ and $a_n$ satisfying assumption \ref{stepsizes} and the matrix of deterministic perturbations $Y$ chosen according to \eqref{yEqn}. In our algorithms, we iteratively choose the perturbations by cycling through columns of $Y$ with period $p+1$ and in steps 2-4, we update the parameters along the direction of estimated gradient according to \eqref{eq:grad-twosided} in the DSPKW-2C algorithm and according to \eqref{eq:grad-onesided} in the DSPKW-1C algorithm. Note the choice of gradient estimate (or the algorithm) is dictated by the simulation budget given to us. The algorithms terminate by returning the 
parameter $\theta_{n_{end}}$ at the end of $n_{end}$ iterations.

\section{Convergence Analysis}
\label{sec:convergenceresults}
In this section we first provide a few lemmas that assist in computing the proposed deterministic perturbation sequence (see \eqref{yEqn} in Section \ref{construct}). In the latter part of the section, we prove the 
almost sure convergence of the iterates for the class of deterministic perturbations characterized in Section \ref{detPerturb}.

The following lemma is useful in obtaining the negative square root of $C$, i.e., $C^{-1/2}$ in a computationally efficient manner. Also note that it takes only $O(p^2)$ operations to compute $C^{-1/2}$
using the lemma and the circulant structure of $C^{-1/2}$. Note that the following lemma could also be utilized in an independent context for efficient computation.
\begin{lemma}
 \label{lemma: gen Sherman-Morrison}
Let $I$ be a $p \times p$ identity matrix and \\
$u=[1,1, \cdots 1]^{T}$
be a $p \times 1$ column vector of 1s, then
$$ (I+uu^T)^{-1/2}= I-\frac{uu^T}{p}+\frac{uu^T}{p\sqrt{(1+p)}}.$$
\end{lemma}
\begin{proof}
It is enough to show that
$$(I+uu^T)\Bigg{[}I-\frac{uu^T}{p}+\frac{uu^T}{p\sqrt{(1+p)}}\Bigg{]}^2=I.$$
Using $\|u\|^2=u^Tu=p$ in the expansion of $\Big{[}I-\frac{uu^T}{p}+\frac{uu^T}{p\sqrt{(1+p)}}\Big{]}^2$
gives the result.
\end{proof}
Let $C$ be defined as in \eqref{eq:C} and $Y=\sqrt{p+1}[C^{-1/2},-C^{-1/2}u].$
Let the perturbations $d_n$ be the columns of $Y.$ 
\begin{lemma}
 The perturbations $d_n$ chosen as columns of Y satisfy properties $\textbf{P1}$ and $\textbf{P2}$.
\end{lemma}
\begin{proof}
 It easily follows from the  discussion in section \ref{construct} on the construction of this specific perturbation sequence. 
\end{proof}

In what follows, we prove the almost sure convergence of the iterates in the DSPKW algorithm (Section \ref{gradEstSec}) under the following assumptions. Note that $\|.\|$ denotes the 2-norm. 
\begin{enumerate}[label= \textbf{A\arabic*.}]
 \item The map $J:\mathbb{R}^p \rightarrow \mathbb{R}$ is Lipschitz continuous and 
 is differentiable with bounded second order derivatives. Further, 
 the map $L:\mathbb{R}^p \rightarrow \mathbb{R}^p$ defined as 
 $L(\theta)=-\nabla J(\theta)$ is Lipschitz continuous.
 \item The step-size sequences $a_n, \delta_n >0, \forall n $  satisfy
 \begin{equation*}\label{stepsizes}
 a_n,\delta_n \rightarrow 0 , \sum_na_n=\infty,
 \sum_n \Big{(}\frac{a_n}{\delta_n}\Big{)}^2<\infty.
 \end{equation*}
 Further, $\frac{a_j}{a_n}\rightarrow 1$ as $n\rightarrow \infty$, for all
 $j \in \{n,n+1,n+2\cdots,n+M\}$ for any given $M>0$ and $b_n=\frac{a_n}{\delta_n}$ is 
 such that $\frac{b_j}{b_n}\rightarrow 1$ as $n\rightarrow \infty$, for all
 $j \in \{n,n+1,n+2,\cdots,n+M\}.$

 \item $\max_n \|d_n\|= K_{0}, \max_n \|D_n\|= K_{1}$. 
 
 \item The iterates $	\theta_n$ remain uniformly bounded almost surely, i.e.,
 $ \sup_n\|\theta_n\|<\infty, \text{ a.s.}$

 \item The ODE $\dot{\theta}(t)=-\nabla J(\theta(t))$ has a compact set 
 $G \subset \mathbb{R}^p$ as its set of asymptotically stable equilibria
 (i.e., the set of local minima of $J$ is compact).
 
 \item The sequences $(M_{n}^{+},\F_n),(M_{n}^{-},\F_n), n\geq0 $ form martingale difference sequences.
 Further, $(M_{n}^{+},M_{n}^{-},n\geq0)$ are square integrable random variables satisfying
 $$\E[\|M_{n+1}^{\pm}\|^2|\F_n]\leq K(1+\|\theta_n\|^2) \text{ a.s., } \forall n\geq0,$$
 for a given constant $K > 0.$
\end{enumerate}
\begin{remark}\label{rem3}
Assumptions \textbf{A1},  \textbf{A2} and \textbf{A5} are standard stochastic approximation conditions. Assumption \textbf{A3} trivially follows from Remark \ref{rem1}.  Assumption \textbf{A4} is the stability condition on the iterates and holds in many applications \cite{spall} (see the discussion in pp 40-41 of \cite{kushcla}).  This condition can also be enforced by projecting the iterates into a compact set, however, the iterates converge to a limiting set that contains all possible limit points (see pp.191 in \cite{kushcla}). Assumption \textbf{A6} gives the condition on the maximum strength of the martingale difference noise under which convergence of the iterates could be ensured and in many stochastic optimization settings this condition could be easily verified using Jensen's inequality and Lipschitz continuity of $\nabla J$ . 
\end{remark}

The following two lemmas aid in the proof of almost sure convergence of the iterates in the DSPKW algorithm.
\begin{lemma}
 Given any fixed integer $P>0$, $\|\theta_{m+k}-\theta_{m}\| \rightarrow 0$ $w.p.1,$ as
 $m \rightarrow \infty,$ for all $k \in \{1,\cdots, P\}.$
\end{lemma}
\begin{proof}
 Fix a $k \in \{1,\cdots,P \}.$ Now
 \begin{align*}
 \begin{split}
 \theta_{n+k} = \theta_n & -\sum_{j=n}^{n+k-1}a_j\Bigg{(}\frac{J(\theta_j+\delta_j d_j)-J(\theta_j-\delta_j d_j)}{2\delta_j}\Bigg{)}d_j \\ 
  &-\sum_{j=n}^{n+k-1}a_jM_{j+1},
 \end{split}
 \end{align*}
 where $M_{j+1}=\frac{(M_{j+1}^{+}-M_{j+1}^{-})d_{j}}{2\delta_j}$. Thus,
 \begin{align*}
 \begin{split}
 \|\theta_{n+k}-\theta_n\| &\leq \sum_{j=n}^{n+k-1}a_j\Bigg{|}\frac{J(\theta_j+\delta_{j} d_j)-J(\theta_j-\delta_{j} d_j)}{2\delta_j}\Bigg{|}\|d_j\|\\
 &+\sum_{j=n}^{n+k-1}a_j\|M_{j+1}\|.
\end{split}
\end{align*}
Now clearly,
$N_n=\sum\limits_{j=0}^{n-1}a_jM_{j+1}, n\geq1,$
forms a martingale sequence with respect to the filtration $\{\F_n\}$.
Further, from the assumption (A6) we have,
\begin{align*}
\sum_{m=0}^{n}\mathbb{E}[\|N_{m+1}-N_{m}\|^2|\mathcal{F}_{m}]& =\sum_{m=0}^{n}\mathbb{E}[a_{m}^2\|M_{m+1}\|^2|\mathcal{F}_{m}]\\
& \leq \sum_{m=0}^{n}a_{m}^2K(1+\|\theta_m\|^2).
\end{align*}
From the assumption (A4), the quadratic variation process of $N_n,n\geq0$ converges 
almost surely. Hence by the martingale convergence theorem, it follows that 
$N_n, n\geq0$ converges almost surely. Hence
$\|\sum\limits_{j=n}^{n+k-1}a_jM_{j+1}\|\rightarrow 0$ almost surely as $n\rightarrow \infty.$
Moreover
\begin{align*}
&\Big{\|}\Big{(}J(\theta_j+\delta_j d_j)-J(\theta_j-\delta_j d_j)\Big{)}d_j\Big{\|} \\
& \leq \Big{|}\Big{(}J(\theta_j+\delta_j d_j)-J(\theta_j-\delta_j d_j)\Big{)}\Big{|}\|d_j\|\\
& \leq K_{0} \Big{(}|J(\theta_j+\delta_j d_j)|+|J(\theta_j-\delta_j d_j)|\Big{)},
\end{align*}
since $\|d_j\|\leq K_{0}, \forall j \geq0.$
Note that
\begin{align*}
|J(\theta_j+\delta_j d_j)|-|J(0)| & \leq|J(\theta_j+\delta_j d_j)-J(0)| \\
& \leq \hat{B} \|\theta_j+\delta_j d_j\|,
\end{align*}
where $\hat{B}$ is the Lipschitz constant of the function $J.$ Hence,
\begin{equation*}
|J(\theta_j+\delta_j d_j)|\leq \tilde{B}(1+\|\theta_j+\delta_j d_j\|),
\end{equation*}
for $\tilde{B}=$max$(|J(0)|,\hat{B}).$ Similarly,
$$|J(\theta_j-\delta_j d_j)|\leq \tilde{B}(1+\|\theta_j-\delta_j d_j\|).$$
From assumption (A1), it follows that
$$\sup_j\Big{\|}\Big{(}J(\theta_j+\delta_j d_j)-J(\theta_j-\delta_j d_j)\Big{)}d_j\Big{\|}\leq \tilde{K}<\infty,$$
for some $\tilde{K}>0.$ Thus,
\newline
$\|\theta_{n+k}-\theta_n\| \leq \tilde{K}\sum\limits_{j=n}^{n+k-1}\frac{a_j}{2\delta_j}+\|\sum_{j=n}^{n+k-1}a_jM_{j+1}\|$
\newline
$\rightarrow 0 \text{ a.s. with } n \rightarrow \infty,$
proving the lemma.
\end{proof}
\begin{lemma}
$\text{ For any } m \geq0,$
$\Big{\|}\sum\limits_{n=m}^{m+P-1}\frac{a_n}{a_m}D_n\nabla J(\theta_n)\Big{\|} \text{ and }$
$\Big{\|}\sum\limits_{n=m}^{m+P-1}\frac{b_n}{b_m}d_nJ(\theta_n)\Big{\|}\rightarrow 0,$
$\text{almost surely, as } m \rightarrow \infty.$
\end{lemma}
\begin{proof}
 From Lemma 3, it can be seen that
  $\|\theta_{m+s}-\theta_{m}\|\rightarrow 0$ as $m\rightarrow \infty,$
 for all $s=1,\cdots,P.$ Also, from assumption (A1), we have
 $\|\nabla J(\theta_{m+s})-\nabla J(\theta_{m})\|\rightarrow 0$ as $m\rightarrow \infty,$
 for all $s=1,\cdots,P.$ Now from Lemma 2, $\sum\limits_{n=m}^{m+P-1}D_n=0$ $\forall m\geq0.$
 Hence $D_m=-\sum\limits_{n=m+1}^{m+P-1}D_n.$ Consider first 
 \begin{align*}
  &\Big{\|}\sum_{n=m}^{m+P-1}\frac{a_n}{a_m}D_n\nabla J(\theta_n)\Big{\|}\\
  & = \Big{\|}\sum_{n=m+1}^{m+P-1}\frac{a_n}{a_m}D_n\nabla J(\theta_n)
  +D_{m}\nabla J(\theta_{m})\Big{\|}\\
  &=\Big{\|}\sum_{n=m+1}^{m+P-1}\frac{a_n}{a_m}D_n \nabla J(\theta_n)
     -\sum_{n=m+1}^{m+P-1}D_n\nabla J(\theta_{m})\Big{\|}\\
  &=\Big{\|}\sum_{n=m+1}^{m+P-1}D_n\Big{(}\frac{a_n}{a_m}\nabla J(\theta_n)
  -\nabla J(\theta_{m})\Big{)}\Big{\|}\\
  &\leq\sum_{n=m+1}^{m+P-1}{\|}D_n{\|}\Big{\|}\Big{(}\frac{a_n}{a_m}\nabla J(\theta_n)
  -\nabla J(\theta_{m})\Big{)}\Big{\|}\\
  &\leq K_{1}\sum_{n=m+1}^{m+P-1}\Big{\|}\Big{(}\frac{a_n}{a_m}-1\Big{)}\nabla J(\theta_n)\Big{\|}
  +\Big{\|}\nabla J(\theta_n)-\nabla J(\theta_{m})\Big{\|}
 \end{align*}
 $\rightarrow 0 \text{ a.s. with } n \rightarrow \infty,$ from assumptions (A1) and (A2).
 Now observe that $\|J(\theta_{m+k})-J(\theta_{m})\|\rightarrow 0$ as $m\rightarrow \infty,$
 for all $k \in \{1,\cdots,P\}$ as a consequence of (A1)
 and Lemma 3. Moreover from $d_m=-\sum\limits_{n=m+1}^{m+P-1}d_n$ we have
\begin{align*}
  & \Big{\|}\sum_{n=m}^{m+P-1}\frac{b_n}{b_m}d_n J(\theta_n)\Big{\|}\\
  & =\Big{\|}\sum_{n=m+1}^{m+P-1}\frac{b_n}{b_m}d_n J(\theta_n)+d_m J(\theta_{m})\Big{\|}\\
  & =\Big{\|}\sum_{n=m+1}^{m+P-1}\frac{b_n}{b_m}d_n J(\theta_n)-\sum_{n=m+1}^{m+P-1}d_n J(\theta_{m})\Big{\|}\\
  & =\Big{\|}\sum_{n=m+1}^{m+P-1}d_n\Big{(}\frac{b_n}{b_m}J(\theta_n)-J(\theta_{m})\Big{)}\Big{\|}\\
  & \leq \sum_{n=m+1}^{m+P-1}\|d_n\|\Big{\|}\Big{(}\frac{b_n}{b_m}J(\theta_n)-J(\theta_{m})\Big{)}\Big{\|}\\
  &\leq K_{0} \sum_{n=m+1}^{m+P-1}\Big{\|}\Big{(}\frac{b_n}{b_m}-1\Big{)} J(\theta_n)\Big{\|}
  +\Big{\|}\Big{(} J(\theta_n)- J(\theta_{m})\Big{)}\Big{\|}
\end{align*}
The claim now follows as a consequence of assumptions (A1) and (A2).
\end{proof}
Finally, using the following theorems, we conclude the analysis by proving the almost sure convergence of the iterates to the set of local minima $G$ of the function $J.$
\begin{theorem}
 $\theta_n, n\geq0$ obtained from DSPKW-2C satisfy $\theta_n \rightarrow G$
 almost surely.
\end{theorem}
\begin{proof}
 Note that
 \begin{equation*}
  \theta_{n+P} = \theta_n-\sum\limits_{l=n}^{n+P-1}a_l\Big{[}\frac{J(\theta_l+\delta_l d_l)-J(\theta_l-\delta_l d_l)}{2\delta_l}d_l+M_{l+1}\Big{]}.
 \end{equation*}
It follows that
 \begin{align*}
  & \theta_{n+P} = \theta_n-\sum_{l=n}^{n+P-1}a_l\nabla J(\theta_l) -\sum_{l=n}^{n+P-1}a_l o(\delta_l) \\
  & -\sum_{l=n}^{n+P-1}a_l(d_ld_l^T-I)\nabla J(\theta_l)-\sum_{l=n}^{n+P-1}a_lM_{l+1}.
  \end{align*}
Now the fourth term on the RHS above can be written as
$$a_n\sum_{l=n}^{n+P-1}\frac{a_l}{a_n}D_{l}\nabla J(\theta_l)=a_n\xi_{n},$$
where $\xi_{n}=o(1)$ from Lemma 4.
Thus, the algorithm is asymptotically analogous to
$$\theta_{n+1}=\theta_n-a_n(\nabla J(\theta_n)+o(\delta)+M_{n+1}).$$
Hence, from Theorem 2 in chapter 2 of \cite{borkar2008stochastic}, it follows that $\theta_n, n\geq0$ converge to a local minima of the function $J.$
\end{proof}
\begin{theorem}
  $\theta_n, n\geq0$ obtained from DSPKW-1C satisfy $\theta_n \rightarrow G$
 almost surely.
\end{theorem}
\begin{proof}
Note that
 \begin{align*}
 \theta_{n+P} = \theta_n-\sum_{l=n}^{n+P-1}a_l\Big{(}\frac{J(\theta_l+\delta_l d_l)}{2\delta_l}\Big{)}d_l-\sum_{l=n}^{n+P-1}a_lM_{l+1}.
 \end{align*} 
It follows that
 \begin{align*}
  \begin{split}
  & \theta_{n+P} = \theta_n- \sum_{l=n}^{n+P-1}a_l\nabla J(\theta_l) 
  -\sum_{l=n}^{n+P-1}a_l\frac{J(\theta_l)}{\delta_l}d_l\\
  &-\sum_{l=n}^{n+P-1}a_l(d_ld_l^T-I)\nabla J(\theta_l)-\sum_{l=n}^{n+P-1}a_lO(\delta_l)\\
  &-\sum_{l=n}^{n+P-1}a_lM_{l+1}.
 \end{split}
 \end{align*}
 Now we observe that the third term on the RHS above is
 \begin{align*}
 & \sum_{l=n}^{n+P-1}a_l\frac{J(\theta_l)}{\delta_l}d_l= \sum_{l=n}^{n+P-1} b_{l}J(\theta_l)d_l\\
 & = b_n\sum_{l=n}^{n+P-1}\frac{b_l}{b_n}\frac{J(\theta_l)}{\delta_l}d_l=b_n\xi^{1}_{n},
 \end{align*}
 where $\xi^{1}_{n}=o(1)$ by Lemma 4. Similarly
 $$\sum_{l=n}^{n+P-1}a_l(d_ld_l^T-I)\nabla J(\theta_l)=a_n\xi^{2}_{n},$$
 with $\xi^{2}_{n}=o(1)$ by Lemma 4.
 The rest follows as in Theorem 5.
\end{proof}

\section{Simulation Experiments}
\label{sec:expts}
In this section, we compare the numerical performance of our DSPKW-2C algorithm against the RDKW algorithm that uses random Bernoulli perturbations and another variant of the RDKW algorithm that uses Hadamard matrix based deterministic perturbations.  We refer them by the acronyms RDKW-2R  and RDKW-2H respectively.  In a similar manner, we also compare DSPKW-1C algorithm against the one-simulation variants RDKW-1R and RDKW-1H. Note that 2 or 1 in the acronyms of these algorithms denote the number of simulations utilized per iteration.\footnote{The implementation is available at \url{https://github.com/cs1070166/1RDSA-2Cand1RDSA-1C/}}

\subsection{Experimental setup}
For the empirical performance evaluation, we consider the following two loss functions:
\paragraph{Quadratic loss}
\begin{align}\label{eq:quadratic}
J(\theta) = \theta\tr A \theta + b\tr \theta. 
\end{align} 
\paragraph{Fourth-order loss}
\begin{align} \label{eq:4thorder}
J(\theta) = \theta \tr A\tr A\theta+0.1 \sum_{j=1}^N (A\theta)^3_j+0.01 \sum_{j=1}^N (A\theta)^4_j.
 \end{align} 
In the loss functions considered above, we set the dimension $p=10$. We choose $A$ such that $pA$ is an upper triangular matrix with each nonzero entry equal to one and $b$ is a $p$-dimensional vector of ones. In our experiments, we follow the same noise assumptions considered in \cite{spall_adaptive}, i.e., for any $\theta$, the additive noise in the objective is given by $[\theta \tr, 1]z$ where $z \sim \N(0,\sigma^2 I_{p+1 \times p+1})$. In all algorithms, we set the step-size schedule as $\delta_n = c/(n+1)^{\gamma}$ and $a_n = 1/(n+B+1)^{\alpha}$ with $\alpha=0.602$ and $\gamma=0.101$. Note that the chosen values for $\alpha$ and $\gamma$ have demonstrated good finite-sample performance empirically, while satisfying the theoretical requirements needed for asymptotic 
convergence (see \cite{spall_adaptive}.  We set the same initial point $\theta_0$ for all the algorithms.

We consider two settings in our experiments. In the first noise-free setting, we do not add any noise to the objective function evaluations and in the second setting, we corrupt the function evaluations by adding noise (with variance parameter $\sigma=0.01$ as described above). We evaluate the performance of these algorithms based on Normalized Mean Square Error (NMSE) metric.  NMSE is defined as the ratio $\l \theta_{n_\text{end}} - \theta^* \r^2 / \l \theta_0 - \theta^*\r^2$, where $\theta_{n_\text{end}}$ is the parameter returned by the algorithm. 

\begin{table}
\centering
\scalebox{0.96}{
\begin{tabular}{|c|c|}
\toprule
\rowcolor{gray!20}
\multicolumn{2}{||c|}{\multirow{2}{*}{\textbf{Noise parameter $\sigma=0$}}}\\[1em]
\midrule
\multirow{1}{*}{ \textbf{Method}} & \textbf{NMSE} \\
\midrule

\textbf{RDKW-2R} &$5.755 \times 10^{-3} \pm 2.460 \times 10^{-3}$ \\
&\\
\textbf{RDKW-2H} &$1.601 \times 10^{-5} \pm 2.724 \times 10^{-20}$ \\ 
&\\
\textbf{DSPKW-2C} &\bm{$2.474 \times 10^{-8} \pm 1.995 \times 10^{-23}$}\\
 \bottomrule

\rowcolor{gray!20}
\multicolumn{2}{||c|}{\multirow{2}{*}{\textbf{Noise parameter $\sigma=0.01$}}}\\[1em]
\midrule
\multirow{1}{*}{ \textbf{Method}} & \textbf{NMSE} \\
\midrule

\textbf{RDKW-2R} &$5.762 \times 10^{-3} \pm 2.473 \times 10^{-3}$ \\
&\\
\textbf{RDKW-2H} &$4.012 \times 10^{-5} \pm 1.654 \times 10^{-5}$\\ 
&\\
\textbf{DSPKW-2C} &\bm{$2.188 \times 10^{-5} \pm 9.908 \times 10^{-6}$}\\

 \bottomrule
\end{tabular}
}
 \caption{NMSE values of two-simulation methods for the quadratic objective
 \eqref{eq:quadratic} without and with noise for 2000 simulations: standard deviation 
 of $100$ replications shown after $\pm$ symbol}
\label{tab:NMSE-quadratic}
\end{table}

\begin{table}
\centering
\scalebox{0.96}{
\begin{tabular}{|c|c|}
\toprule
\rowcolor{gray!20}
\multicolumn{2}{||c|}{\multirow{2}{*}{\textbf{Noise parameter $\sigma=0$}}}\\[1em]

\midrule
\multirow{1}{*}{ \textbf{Method}} & \textbf{NMSE} \\
\midrule

\textbf{RDKW-2R} &$2.747 \times 10^{-2} \pm 1.413 \times 10^{-2}$ \\
&\\
\textbf{RDKW-2H} &$3.901 \times 10^{-3} \pm 4.359 \times 10^{-18}$ \\ 
&\\
\textbf{DSPKW-2C} &\bm{$3.535 \times 10^{-3} \pm 1.743 \times 10^{-18}$}\\
 \bottomrule

\rowcolor{gray!20}
\multicolumn{2}{||c|}{\multirow{2}{*}{\textbf{Noise parameter $\sigma=0.01$}}}\\[1em]
\midrule
\multirow{1}{*}{ \textbf{Method}} & \textbf{NMSE} \\
\midrule

\textbf{RDKW-2R} &$2.762 \times 10^{-2} \pm 1.415 \times 10^{-2}$ \\
&\\
\textbf{RDKW-2H} &$3.958 \times 10^{-3} \pm 4.227 \times 10^{-4}$\\ 
&\\
\textbf{DSPKW-2C}& \bm{$3.598 \times 10^{-3} \pm 4.158 \times 10^{-4}$}\\

 \bottomrule
\end{tabular}
}
 \caption{NMSE values of two-simulation methods for the fourth order
 objective \eqref{eq:4thorder} without and with noise for 10000 simulations: standard deviation 
 of $100$ replications shown after $\pm$ symbol}
\label{tab:NMSE-fourthorder}
\end{table}


\begin{table}
\centering
\scalebox{0.96}{
\begin{tabular}{|c|c|}
\toprule
\rowcolor{gray!20}
\multicolumn{2}{||c|}{\multirow{2}{*}{\textbf{Noise parameter $\sigma=0$}}}\\[1em]

\midrule
\multirow{1}{*}{ \textbf{Method}} & \textbf{NMSE} \\
\midrule

\textbf{RDKW-1R} &$8.584 \times 10^{-2} \pm 3.681 \times 10^{-2}$ \\
&\\
\textbf{RDKW-1H} &$2.770 \times 10^{-2} \pm 3.836 \times 10^{-17}$ \\ 
&\\
\textbf{DSPKW-1C} &\bm{$8.225 \times 10^{-3} \pm 1.569 \times 10^{-17}$}\\
 \bottomrule

 \rowcolor{gray!20}
\multicolumn{2}{||c|}{\multirow{2}{*}{\textbf{Noise parameter $\sigma=0.01$}}}\\[1em]
\midrule
\multirow{1}{*}{ \textbf{Method}} & \textbf{NMSE} \\
\midrule

\textbf{RDKW-1R} &$8.582 \times 10^{-2} \pm 3.691 \times 10^{-2}$ \\
&\\
\textbf{RDKW-1H} &$2.774 \times 10^{-2} \pm 2.578 \times 10^{-4}$\\ 
&\\
\textbf{DSPKW-1C} &\bm{$8.225 \times 10^{-3} \pm 5.959 \times 10^{-5}$}\\
 \bottomrule

\end{tabular}
}
 \caption{NMSE values of one-simulation methods for the quadratic
 objective \eqref{eq:quadratic} without and with noise for 20000 simulations: standard deviation 
 of $100$ replications shown after $\pm$ symbol}
\label{tab:NMSE-quadratic-1sim}
\end{table}

\begin{table}
\centering
\scalebox{0.96}{
\begin{tabular}{|c|c|}
\toprule
\rowcolor{gray!20}
\multicolumn{2}{||c|}{\multirow{2}{*}{\textbf{Noise parameter $\sigma=0$}}}\\[1em]

\midrule
\multirow{1}{*}{ \textbf{Method}} & \textbf{NMSE} \\
\midrule

\textbf{RDKW-1R} &$3.192 \times 10^{-1} \pm 1.991 \times 10^{-1}$ \\
&\\
\textbf{RDKW-1H} &$8.173 \times 10^{-2} \pm 1.255 \times 10^{-16}$ \\ 
&\\
\textbf{DSPKW-1C} &\bm{$4.403 \times 10^{-2} \pm 9.066 \times 10^{-17}$}\\

 \bottomrule

\rowcolor{gray!20}
\multicolumn{2}{||c|}{\multirow{2}{*}{\textbf{Noise parameter $\sigma=0.01$}}}\\[1em]
\midrule
\multirow{1}{*}{ \textbf{Method}} & \textbf{NMSE} \\
\midrule

\textbf{RDKW-1R} & $3.240 \times 10^{-1} \pm 1.836 \times 10^{-1}$ \\
&\\
\textbf{RDKW-1H} &$8.916 \times 10^{-2} \pm 1.896 \times 10^{-2}$\\ 
&\\
\textbf{DSPKW-1C} &\bm{$4.972 \times 10^{-2} \pm 9.812 \times 10^{-3}$}\\
 \bottomrule
\end{tabular}}
 \caption{NMSE values of one-simulation methods for the fourth order
 objective \eqref{eq:4thorder} without and with noise for 20000 simulations: standard deviation  
of of $100$ replications shown after $\pm$ symbol}
\label{tab:NMSE-fourthorder-1sim}
\end{table}

\subsection{Discussion of Results}
The performance comparisons of all the algorithms based on NMSE values are summarized in Tables \ref{tab:NMSE-quadratic}, \ref{tab:NMSE-fourthorder}, \ref{tab:NMSE-quadratic-1sim} and \ref{tab:NMSE-fourthorder-1sim}. In the tables, we have highlighted the algorithm that has the minimum NMSE.  We summarize our findings:
\begin{itemize}
\item Even in the absence of noise, due to the random directions chosen by RDKW-2R and RDKW-1R algorithms, the standard deviation is significantly high compared to the corresponding deterministic counterparts. 
\item We would like to emphasize that the quality of the solution (characterized by standard deviation) is significantly better for the case of proposed deterministic perturbations compared to the existing Hadamard based deterministic perturbations and random perturbations.  Note however that we do not make comparisons between two-simulation and one-simulation algorithms.  
\item In the case of two simulation algorithms (see Tables \ref{tab:NMSE-quadratic} and \ref{tab:NMSE-fourthorder}), DSPKW-2C performs marginally better than RDKW-2H, while both of them outperform RDKW-2R significantly.
\item In the case of one simulation algorithms (see Tables \ref{tab:NMSE-quadratic-1sim} and \ref{tab:NMSE-fourthorder-1sim}), DSPKW-1C performs better than both RDKW-1H and RDKW-1R.
\end{itemize}

\section{Conclusions}
\label{sec:conclusions}
We have generalized the deterministic perturbation sequences from lexicographical ordering and Hadamard matrix based constructions for the RDKW algorithm and presented a novel construction of deterministic perturbations that has least cycle length within the class of deterministic perturbation sequences. Further, we have proved the almost sure convergence of the iterates for the class of deterministic perturbation sequences. Now that we have a characterization of the class of deterministic perturbation 
sequences, it would be interesting as future work, to theoretically study and compare the rate of convergence of deterministic perturbation algorithms against their random perturbation counterparts. A challenging future direction would be to study the asymptotic normality or weak convergence of the iterates. It would also be interesting to similarly characterize the class of deterministic perturbation sequences for the SPSA algorithm.

\bibliography{reference}
\bibliographystyle{IEEEtran}
\end{document}